\documentclass[a4paper,11pt]{article}

\usepackage{blindtext,titlefoot}

\usepackage{amsfonts,amsthm,amssymb,amsmath}
\usepackage[title,titletoc,toc]{appendix}
\usepackage[utf8]{inputenc}
\usepackage[affil-it]{authblk}
\usepackage{xcolor}
\usepackage{tikz}
\newcommand\COMP{\hbox{C\kern -.58em {\raise .54ex \hbox{$\scriptscriptstyle |$}}
\kern-.55em {\raise .53ex \hbox{$\scriptscriptstyle |$}} }}
\newcommand\NN{\hbox{I\kern-.2em\hbox{N}}}
\newcommand\RR{\hbox{I\kern-.2em\hbox{R}}}
\newcommand\sRR{{\it \hbox{I\kern-.2em\hbox{R}}}}
\newcommand\QQ{\hbox{I\kern-.53em\hbox{Q}}}
\newcommand\PP{\hbox{I\kern-.53em\hbox{P}}}
\newcommand\EE{\hbox{I\kern-.53em\hbox{E}}}
\newcommand\ZZ{{{\rm Z}\kern-.28em{\rm Z}}}
\newcommand\be{\begin{equation}}
\newcommand\ee{\end{equation}}
\newtheorem{theorem}{Theorem}[section]

\newtheorem{remark}[theorem]{Remark}

\newtheorem{lemma}[theorem]{Lemma}

\newtheorem{definition}[theorem]{Definition}

\makeatletter
\newcommand*\bigcdot{\mathpalette\bigcdot@{.5}}
\newcommand*\bigcdot@[2]{\mathbin{\vcenter{\hbox{\scalebox{#2}{$\m@th#1\bullet$}}}}}
\makeatother

\numberwithin{equation}{section}

\begin{document}

\title{log-optimal portfolio without NFLVR:  existence, complete characterization,  and duality\thanks{This research is supported by  the Natural Sciences
and Engineering Research Council of Canada, through Grant RES0020459}}

\author{Tahir Choulli and Sina Yansori}

\affil{\small{Department of Mathematical and Statistical Sciences, University of Alberta, Edmonton, Canada}}

% \date{}

\maketitle

\begin{abstract}
% text of abstract goes here!
This paper addresses the log-optimal portfolio for a general semimartingale model.  The most advanced literature on the topic elaborates  existence and characterization of this portfolio under no-free-lunch-with-vanishing-risk assumption (NFLVR). There are many financial models violating NFLVR, while admitting the log-optimal portfolio on the one hand. On the other hand, for financial markets  under progressively enlargement of filtration, NFLVR remains completely an open issue, and hence the literature can be applied to these models. Herein, we provide a complete characterization of log-optimal portfolio and its associated optimal deflator, necessary and sufficient conditions for their existence, and we elaborate their duality as well  without NFLVR.  \end{abstract}

\section{Introduction}

Since the seminal papers of Merton \cite{merton71,merton73}), the theory of utility maximization and optimal portfolio has been developed successfully in many directions and in different frameworks. These achievements can be found in \cite{Karatzas,KW99,CSW,KZ}, and the references therein to cite few. Besides the Markowitz' portfolio, thanks to the nice properties of the logarithm utility, the log-optimal portfolio draw tremendous attention since a while. This resulted in a large literature on the topic for different level of generalities. The most advanced of this literature, under the assumption of no-free-lunch-with-vanishing-risk assumption (NFLVR herefater), provides explicit characterization for this optimal portfolio for the general semimartingale market models, see \cite{ChristensenLarsen2007,GollKallsen, HulleySchweizer} and the references therein to cite few. However, there are many financial models that violates NFLVR, while they might admit the log-optimal portfolio, see \cite{ChoulliDengMa,Lowenstein,Ruf}, and hence these assert somehow that NFLVR might be too strong for a financial market model to be ``acceptable and worthy". For market models under progressive enlargement of filtration, which incorporate the two important settings of credit risk and life insurance, NFLVR remains an open issue, and hence the existing literature is not applicable to these models on the one hand. On the other hand, for these latter market models,  the no-unbounded-profit-with-bounded-risk received full attention as it is the minimal no-arbitrage condition for a market model to be financial and quantitatively ``acceptable" and viable, see \cite{ACDJ1, ACDJ2} and the references therein. Furthermore, recently there has been an interest to extend the existing results on the optimal portfolio and utility maximization without NFLVR, see \cite{Fontana}. 

This paper contains four sections including the current one. Section  \ref{section2} presents the mathematical model and notation that we work with, and provides its preliminary analysis.  Section \ref{section3} elaborates the main result of the paper and discuss its relationship to the literature, while Section \ref{section4} proves the main theorem.  The paper contains an appendix where some proofs are relegated and  some useful existing results are recalled.

%%%%%%%%%%%%%%%%%%%%%%%%%%%%%%%%%%%%%%%%%%%%%%%%%%%%%%
%%%%%%%%%%%%%%%%%%%%%%%%%%%%%%%%%%%%%%%%%%%%%%%%%%%%%%
%%%%%%%%%%%%%%%%%%%%%%%%%%%%%%%%%%%%%%%%%%%%%%%%%%%%%%%%%%%
%%%%%%%%%%%%%%%%%%%%%%%%%%%%%%%%%%%%%%%%%%%%%%%%%%%%%%%
%%%%\subsection{Notation and Preliminaries
%%%%%%%%%%%%%%%%%%%%%%%%%%%%%%%%%%%%%
%%%%%%%%%%%%%%%%%%%%%%%%%%%%%%%%%%%%%%%%%%%%%%%%%%%%%%%%%%%%%%%%%%%%%%%%%%
%%%%%%%%%%%%%%%%%%%%%%%%%%%%%%%%%SUBSECTION OF MAIN RESULTS
%%%%%%%%%%%%%%%%%%%%%%%%%%%%%%%%%%%%%%%%%%%%%%%%%%%%%%%%%%%%%%%%%%%%%%%%%%
%%%%%%%%%%%%%%%%%%%%%%%%%%%%%%%%%%%%%%%%%%%%%%%%%%%%%%%%%%%%%%%%%%%%%%%%%
%\section{Log-optimal portfolio: Duality and charaterization}\label{DeepResultonDual}

\section{The mathematical model, notations and preliminaries}\label{section2} 
%%%%%%%%%%%%%%%%%%%%%%%%%%%%%%%%%%%%%%%%%%%%%%%%%%%%%%%%%%%%%%%%%%%
Throughout the paper, we consider a filtered probability space $\left(\Omega, {\cal F}, \mathbb H:=({\cal H}_t)_{t\geq 0},P\right)$ satisfying the usual conditions of right continuity and completeness.  On this stochastic basis, we suppose given a $d$-dimensional semimartingale, $X$, that represents the discounted price process of $d$ risky assets. Throughout the paper, the set ${\cal M}( Q)$ denotes the set of all  martingales under $Q$, while ${\cal A}( Q)$ (respectively ${\cal A}^+(Q)$) denotes the set of all  optional processes with integrable variation (respectively nondecreasing and integrable) under $Q$. When $Q=P$, we simply omit the probability for the sake of simple notations.  For a semimartingale $Y$, by $L(Y)$ we denote the set of predictable processes that are $X$-integrable in the semimartingale sense.  For $\varphi\in L(Y)$, the resulting integral of $\varphi$ with respect to $X$ is denoted by $\varphi\cdot Y$. For any local martingale $M$, we denote by $L^1_{loc}(M)$ the set $\mathbb H$-predictable processes $\varphi$ that are $Y$-integrable and the resulting integral $\varphi\cdot M$ is a local martingale. If ${\cal C}$ is the set of processes, then ${\cal C}_{loc}$ is the set of processes, $Y$, for which there exists a sequence of stopping times, $(T_n)_{n\geq 1}$, that increases to infinity and $Y^{T_n}$ belongs to ${\cal C}$, for each $n\geq 1$. For any semimartinagle, $L$, we denote by ${\cal E}(L)$ the Doleans-Dade (stochastic) exponential, it is the unique solution to the stochastic differential equation 
$$dY=Y_{-}dL,\quad X_0=1,\quad\mbox{and is given by}\quad {\cal E}_t(L)=\exp(L_t-{1\over{2}}\langle L^c\rangle_t)\prod_{0<s\leq t}(1+\Delta L_s)e^{-\Delta L_s}.$$
%%%%%%%%%%%%%%%%%%%%%%%%%%%%%%%%%%%%%%%%%%%%%%%%%%%%%%
%%%%For any process $Y$, the $\mathbb H$-optional projection and dual optional projection of $Y$, when they exist, will be denote by $^{o,\mathbb H}Y$ and $Y^{o,\mathbb H}$ respectively.  Similarly, we denote by $^{p,\mathbb H}Y$ and $Y^{p,\mathbb H}$ the $\mathbb H$-predictable projection and dual predictable projection of $X$ when they exist. 
%%%%%%%%%%%%%%%%%%%%%%%%%%%%%%%%%%%%%%%%%%%%%%%%%%%%%%
In the following, we recall the predictable characteristics of $X$ that will play important role in the results' statement and their proofs as well. This requires some definitions and notations that we start introducing. On $\Omega\times[0,+\infty)\times{\mathbb R}^d$,   we consider  
$$\label{sigmaFields}
\widetilde {\cal O}(\mathbb H):={\cal O}(\mathbb H)\otimes {\cal
B}({\mathbb R}^d),\ \ \ \ \ \widetilde{\cal P}(\mathbb H):= {\cal
P}(\mathbb H)\otimes {\cal B}({\mathbb R}^d),
$$
where ${\cal B}({\mathbb R}^d)$ is the Borel $\sigma$-field on
${\mathbb R}^d$, the optional and predictable $\sigma$-fields respectively.  To $X$, we associate the optional random measure $\mu$
defined by
\begin{eqnarray*}\label{mesuresauts}
\mu(dt,dx):=\sum_{u>0} I_{\{\Delta X_u \neq 0\}}\delta_{(u,\Delta
X_u)}(dt,dx)\,.\end{eqnarray*}
%%%%%%%%%%%%%%%%%%%%%%%%%%%%%%%%%%%%%%%
For a product-measurable functional
$W\geq 0$ on $\Omega\times \mathbb R_+\times{\mathbb R}^d$, we
denote $W\star\mu$ (or sometimes, with abuse of notation,
$W(x)\star\mu$) the process
\begin{eqnarray*}\label{Wstarmu}
(W\star\mu)_t:=\int_0^t \int_{{\mathbb R}^d\setminus\{0\}}
W(u,x)\mu(du,dx)=\sum_{0<u\leq t} W(u,\Delta X_u) I_{\{ \Delta
X_u\not=0\}}.\end{eqnarray*}

We define,   on
$\Omega\times\mathbb R_+\times{\mathbb R}^d$, the
measure $M^P_{\mu}:=P\otimes\mu$ by $$M^P_{\mu}\left(W\right):=\int W
dM^P_{\mu}:=E\left[(W\star\mu)_\infty\right],$$ (when the expectation is well defined). The \emph{conditional expectation} given $
\widetilde{\cal P}(\mathbb H)$ of a product-measurable functional
$W$, denoted by $M^P_{\mu}(W|\widetilde{\cal P}(\mathbb H))$, is the unique $ \widetilde{\cal P}(\mathbb H)$-measurable
functional $\widetilde W$ satisfying
$$
E\left[(W I_{\Sigma}\star\mu)_\infty \right]=E\left[({\widetilde W}
I_{\Sigma}\star\mu)_\infty \right]\ \ \ \mbox{for all}\
\Sigma\in\widetilde{\cal P}(\mathbb H).$$
%%%%%%%%%%%%%%%%%%%%%%%%%%%%%%%%%%%%%%%%%%%%%%%%%%%
For the reader's convenience, we recall {\it the
 canonical decomposition} of $X$ (for more related details, we refer the reader to \cite[Theorem 2.34, Section II.2]{JS03})
\begin{equation}\label{modelSbis}
X=X_0+X^c+h\star(\mu-\nu)+b\cdot A+(x-h)\star\mu,\end{equation}
where $h$, defined as $h(x):=xI_{\{ \vert x\vert\leq 1\}}$, is the
truncation function, and \mbox{$h\star(\mu-\nu)$} is the unique pure jump $\mathbb H$-local martingale with jumps given by $h(\Delta X)I_{\{\Delta X\not=0\}}$. For the matrix $C$ with entries $C^{ij}:=\langle X^{c,i},
X^{c,j}\rangle $, and $\nu$, we can find a version satisfying
$$  C=c\cdot A,\ \nu(d t,\ d x)=d A_t F_t(d x),\
F_t(\{0\})=0,\ \displaystyle{\int} (\vert
x\vert^2\wedge 1)F_t(d x)\leq 1.
$$ Here $A$ is increasing and continuous due to the quasi-left-continuity of $X$, $b$
and $c$ are predictable processes,
 $F_t(d x)$ is a predictable kernel, $b_t(\omega)$ is a vector in $\hbox{I\kern-.18em\hbox{R}}^d$ and
$c_t(\omega)$ is a symmetric $d\times d$-matrix, for all $(\omega,\
t)\in\Omega\times \mathbb R_+$. The quadruplet
\begin{eqnarray*}\label{FpredictCharac}
(b,c, F, A)\ \mbox{are the predictable characteristics of}\ X.\end{eqnarray*}
 For more details about these {\it predictable characteristics} and other
related issues, we refer to \cite[Section II.2]{JS03}. For the sake of simplicity, we will consider models, that we call {\it $\sigma$-special}, defined as follows. 
\begin{definition}\label{SigmaSpecial}  The model $(X,\mathbb H)$ is called $\sigma$-special if there exists a real-valued and $\mathbb H$-predictable process $\varphi$ such that 
\begin{eqnarray}\label{sigmaSpecial}
0<\varphi\leq 1\quad \mbox{and}\quad \sum \varphi\vert\Delta X\vert I_{\{\vert\Delta X\vert>1\}}\in {\cal A}^+_{loc}(\mathbb H).
\end{eqnarray}
\end{definition}
It is clear that (\ref{sigmaSpecial}) is equivalent to $\int_{(\vert x\vert>1)}\vert x\vert F(dx)<+\infty$ $P\otimes A$-a.e., or to $\varphi\cdot X$ being a special semimartingale (i.e. $\sup_{0<s\leq\cdot}\vert \varphi\Delta X_s\vert\in{\cal A}^+_{loc}$), or equivalently . If $X$ is locally bounded, then it is $\sigma$-special.
%%%%%%%%%%%%%%%%%%%%%%%%%%%%%%%%%%%%%%%%%%%%%%%%%%%%%%%%%%%%%
\section{Main Result}\label{section3} 
%%%%%%%%%%%%%%%%%%%%%%%%%%%%%%%%%%%%%%%%%%%%%%%%%%%%%%%%%%%%%%%%
%%%%%%%%%%%%%%%%%%%%%%%%%%%%%%%%%%%%%%%%%%%%%%%%%%%%%%%%%%%%%%%%
This section states the main theorem of the paper, and discusses its relationship to the literature. To this end, we recall the set of admissible portfolios and deflators. Throughout the paper, we denote $\Theta(X,\mathbb H)$ the following set 
\begin{eqnarray}\label{LogOptimization} 
\Theta(X,\mathbb H):=\Bigl\{\theta\in L(X,\mathbb H)\big|\quad E\left[\max(0,-\ln( 1+(\theta\cdot X)_T)\right]<+\infty\Bigr\}.
\end{eqnarray}

\begin{definition}\label{DeflatorDefinition} Let $Z$ be a process. $Z$ is called a  deflator for $(X,\mathbb H)$ if $Z>0$ and $Z{\cal E}(\varphi\cdot X)$ is an $\mathbb H$-supermartingale, for any $\varphi\in L(X, \mathbb H)$ such that $\varphi\Delta X\geq -1$.\\ Throughout the paper, the set of all deflator for  $(X,\mathbb H)$ will be denoted by ${\cal D}(X,\mathbb H)$.
\end{definition}
%%%%%%%%%%%%%%%%%%%%%%%%%%%%%%%%%%%%%%%%%%%%%%%%%%%%%%%%%%%%%
\begin{theorem}\label{LemmaCrucial}
Suppose $(X,\mathbb H)$ is $\sigma$-special and quasi-left-continuous with predictable characteristics $\left(b,c,F, A\right)$. Then the following assertions are equivalent.\\
{\rm{(a)}} The set ${\cal D}_{log}(X,\mathbb H)$, given by 
\begin{eqnarray}\label{DsetLog}
{\cal D}_{log}(X,\mathbb H):=\left\{Z\in {\cal D}(X,\mathbb H)\quad \big|\quad E[-\ln(Z_T)]<+\infty\right\},
\end{eqnarray}
is not empty (i.e. ${\cal D}_{log}(X,\mathbb H)\not=\emptyset$).\\
%%%\begin{eqnarray}\label{C2forX}
%%%\left(\int\vert \ln(1+\theta^{tr}x)-\theta^{tr}h(x)\vert F^X(dx)\right)\is A_T<+\infty\quad P\mbox{-a.s.}
%%%\end{eqnarray}\ E\left[ h^{(0)}(\widetilde N, P)_T+{\widetilde V}_T\right]<+\infty
{\rm{(b)}} There exists an $\mathbb H$-predictable process $\widetilde\varphi\in{\cal L}(X,\mathbb H)$ such that, for any $\varphi$ belonging to  ${\cal L}(X,\mathbb H)$,  the following hold 
\begin{eqnarray}
&&E\left[{\widetilde V}_T+{1\over{2}}(\widetilde\varphi^{tr}c\widetilde\varphi\cdot A)_T+(\int({{ -\widetilde\varphi^{tr}x}\over{1 +\widetilde\varphi^{tr}x }} +\ln(1 +\widetilde\varphi^{tr}x)) F(dx)\cdot A)_T\right]<+\infty ,\label{Condi11}\\
&& {\widetilde V}:=\Big\vert \widetilde\varphi^{tr}(b-c\widetilde\varphi)+\int \left[{{\widetilde\varphi^{tr}x}\over{1+\widetilde\varphi^{tr}x}}-\widetilde\varphi^{tr}h(x)\right] F(dx)\Big\vert\cdot A,\label{processV}\\
&&(\varphi-\widetilde\varphi)^{tr}(b-c\widetilde\varphi)+ \int \left( {{(\varphi-\widetilde\varphi)^{tr}x}\over{1+{\widetilde\varphi}^{tr}x}}-(\varphi-\widetilde\varphi)^{tr}h(x)\right)F(dx)\leq 0. \label{C6forX}
\end{eqnarray}
{\rm{(c)}} There exists a unique $\widetilde Z\in{\cal D}(X,\mathbb H)$ such that 
\begin{eqnarray}\label{dualSolution}
\inf_{Z\in{\cal D}(X,\mathbb H)}E[-\ln(Z_T)]=E[-\ln(\widetilde Z_T)]<+\infty.
\end{eqnarray}
{\rm{(d)}} There exists a unique $\widetilde\theta\in\Theta(X,\mathbb H)$ such that 
\begin{eqnarray}\label{PrimalSolution}
\sup_{\theta\in\Theta(X,\mathbb H)}E[\ln(1+(\theta\cdot X)_T)]=E[\ln(1+(\widetilde\theta\cdot X)_T)]<+\infty.
\end{eqnarray}
Furthermore,  when these assertions hold, the following hold. \begin{eqnarray}
&& \widetilde\varphi\in L(X^c,\mathbb H)\cap {\cal L}(X,\mathbb H),\quad \sqrt{((1+\widetilde\varphi^{tr}x)^{-1}-1)^2\star\mu}\in{\cal A}^+_{loc}(\mathbb H),\label{integrabilities}\\
&&{1\over{\widetilde Z}}={\cal E}(\widetilde\varphi\cdot X),\ \widetilde Z:={\cal E}(K-{\widetilde V}),\ K:=\widetilde\varphi\cdot X^c+{{-\widetilde\varphi^{tr}x}\over{1+\widetilde\varphi^{tr}x}}\star(\mu-\nu).\hskip 1cm\label{duality}\\
&&\widetilde\varphi=\widetilde\theta (1+(\widetilde\theta\cdot X)_{-})^{-1}\quad \rm{and}\quad \widetilde\theta=\widetilde\varphi{\cal E}_{-}(\widetilde\varphi\cdot X)\quad P\otimes A\rm{-a.e.}.\label{equalityFiTheta}\end{eqnarray}
\end{theorem} 
It is important to notice that, for any $Z\in {\cal D}(X,\mathbb H)$, we always have $E\ln^+(Z_T)\leq \ln(2)$. Furthermore, one can easily prove that the following two assertions are equivalent:\\
(a) $Z\in {\cal D}_{log}(X,\mathbb H)$ (i.e. $-\ln(Z_T)$ is integrable or equivalently  $(\ln(Z_T))^-$ is integrable),\\
(b) $\{-\ln(Z_t),\ 0\leq t\leq T\}$, or equivalently $\{(\ln(Z_t))^-,\ 0\leq t\leq T\}$, is uniformly integrable submartingale.\\ 
Besides this, for a positive local martingale $Z$, the condition $E[-\ln(Z_T)]<+\infty$ does not guarantee  that this $Z$ is a martingale, while it implies that $K:=Z_{-}^{-1}\cdot Z$ is a  martingale satisfying $E[\sup_{0\leq t\leq T}\vert K_t\vert]<+\infty$ instead, see Lemma \ref{H0toH1martingales} for this latter fact. As a result of this discussion, we conclude that Theorem \ref{LemmaCrucial} extends deeply the existing literature on the log-optimal portfolio by dropping the no-free-lunch-with-vanishing-risk condition on the model. This assumption is really a vital assumption for the analysis of \cite{GollKallsen}. This achievement is due to our approach that differs fundamentally from that of \cite{GollKallsen}, while it is inspired from the approach of \cite{ChoulliStricker2007} with a major difference. This difference lies in dropping all assumptions on the model $(X,\mathbb H)$ considered in \cite{ChoulliStricker2007}, which  guarantee that the minimizer of a functional belongs to the interior of its effective domain. We recall our aforementioned claim that the ``$\sigma$-special assumption" for $(X,\mathbb H)$ is purely technical and is not related at all to the minimizer of this functional.  In conclusion, our theorem establishes the duality, under basically no assumption, besides describing the optimal dual solution when it exists as explicit as possible. This, furthermore, proves that in general, this optimal deflator might not be a local martingale deflator.

%%%%%%%%%%%%%%%%%%%%%%%%%%%%%%%%%%%%%%%%%%%%%%%%%%%%
%%%%%%%%%%%%%%%%%%%%%%%%%%%%%%%%%%%%%%%%%%%%%%%%%%%
\begin{remark}\label{remark3.3} It is clear that the process $V$ is well defined. This is due to  
\begin{eqnarray*}
\int [{{\widetilde\varphi^{tr}x}\over{1+\widetilde\varphi^{tr}x}}-\widetilde\varphi^{tr}h(x)] F(dx)&&=-\int_{(\vert x\vert\leq 1)} {{(\widetilde\varphi^{tr}x)^2}\over{1+\widetilde\varphi^{tr}x}}F(dx)\\
&&-\int_{(\vert x\vert> 1)} {{1}\over{1+\widetilde\varphi^{tr}x}}F(dx) +F(\vert x\vert >1),
 \end{eqnarray*}
which is a well defined integral with values in $[-\infty, +\infty)$. \\
 Similarly for the LHS term of (\ref{C6forX}), the integral term is well defined for any $\varphi\in{\cal L}(X,\mathbb H)$. Indeed, due to $\Omega\times [0,+\infty)=\cup_{n\geq 0}(\vert \varphi\vert\leq n)$, for any process $\varphi\in{\cal L}(X,\mathbb H)$, it is enough to prove that the integral term is well defined for bounded  $\varphi\in{\cal L}(X,\mathbb H)$. To this end, on the one hand, we write
 \begin{eqnarray*}
  \int \left( {{\varphi^{tr}x}\over{1+{\widetilde\varphi}^{tr}x}}-\varphi^{tr}h(x)\right)F(dx)&&=-\int_{(\vert x\vert\leq 1)} {{(\varphi^{tr}x)({\widetilde\varphi}^{tr}x)}\over{1+{\widetilde\varphi}^{tr}x}}F(dx)-F(\vert x\vert>1)\\
  &&\hskip -1cm+ \int _{(\vert x\vert> 1)}  {{\varphi^{tr}x+1}\over{1+{\widetilde\varphi}^{tr}x}}F(dx)+\int _{(\vert x\vert> 1)} {{{\widetilde\varphi}^{tr}x}\over{1+{\widetilde\varphi}^{tr}x}}F(dx).
  \end{eqnarray*}
  On the other hand, since $\varphi$ is $X$-integrable (as it is bounded), both processes $I_{\{\vert \Delta X\vert\leq 1\}}\cdot [K, \varphi\cdot X]$  and $[\sum I_{\{\vert\Delta X\vert>1\}},K]$ have  locally integrable variations and their compensators are  $$-\int_{(\vert x\vert\leq 1)} {{(\varphi^{tr}x)({\widetilde\varphi}^{tr}x)}\over{1+{\widetilde\varphi}^{tr}x}}F(dx)\cdot A\ \mbox{ and}\ -\int _{(\vert x\vert> 1)} {{{\widetilde\varphi}^{tr}x}\over{1+{\widetilde\varphi}^{tr}x}}F(dx)\cdot A$$ respectively. This proves that the integral is well defined with values in $(-\infty, +\infty]$.
  \end{remark}

%% \subsection{Particular cases}
 %%%%%%%%%%%%%%%%%%%%%%%%%%%%%%%%%%%%%%%%%%%%%%%%%%%%%%%%%%%%
 %%%%%%%%%%%%%%%%%%%%%%%%%%%%%%%%%%%%%%%%%%%%%%%%%%%%%%%%%%%%%
 \section{Proof of Theorem \ref{LemmaCrucial}}\label{section4} 
 
\begin{proof} {\it of Theorem \ref{LemmaCrucial}.} It is clear that (c)$\Longrightarrow$(a) is  obvious, and hence the proof of the theorem reduces to proving   (a)$\Longrightarrow$(b)$\Longrightarrow$(c), (b)$\Longrightarrow$(d), (d)$\Longrightarrow$(a), and as long as assertion (b) holds the properties in (\ref{integrabilities})-(\ref{duality}) hold also. Thus, the rest of this proof is divided into three steps. The first step proves that assertion (b) implies both assertions (c) and (d) and (\ref{integrabilities})-(\ref{duality}). The second step deals with (d)$\Longrightarrow$ (a), while the third step addresses (a)$\Longrightarrow$ (b). \\
%%%%%%%%%%%%%%%%%%%%%%%%%%%%%%%%%%%%%%%%%%%%%%%%%%%%%%%%%%%%%%
{\bf Step 1.} Here, we assume that assertion (b) holds, and focus on proving assertions (c) and (d), and  (\ref{integrabilities})-(\ref{duality}).  Then due to (\ref{Condi11}) and  
$$(1+y)\ln(1+y)-y\geq {{1-\delta}\over{2}}{{y^2}\over{1+y}} I_{\{\vert y\vert\leq \delta\}}+ {{\delta}\over{2(1+\delta)}}\vert y\vert I_{\{\vert y\vert> \delta\}}$$
for any $\delta\in(0,1)$ and any $y\geq -1$, we deduce that for $\delta\in (0,1)$ the following  
\begin{eqnarray*}
&&\widetilde\varphi ^{tr}c\widetilde\varphi \cdot A,\quad \int_{\mathbb R^d\setminus\{0\}} \left({{\widetilde\varphi^{tr} x }\over{1 +\widetilde\varphi^{tr} x }} \right)^2 I_{\{\vert \widetilde\varphi^{tr} x\vert\leq\delta\}}F(dx)\cdot A,\\
&&\mbox{and}\quad \int_{\mathbb R^d\setminus\{0\}} {{\vert \widetilde\varphi^{tr} x \vert }\over{1 +\widetilde\varphi^{tr} x }}  I_{\{\vert \widetilde\varphi^{tr} x\vert>\delta\}}F(dx)\cdot A
\end{eqnarray*}
 are integrable processes, and hence $ \widetilde\varphi\in L(X^c,\mathbb H)$ and $ \sqrt{((1+\widetilde\varphi^{tr}x)^{-1}-1)^2\star\mu}\in{\cal A}^+_{loc}$ due to Lemma \ref{lemm4F-Gintregrability} (see Appendix A). Hence  $K$, defined in (\ref{duality}), is a well defined local martingale satisfying $\Delta K+1=(1+\widetilde\varphi^{tr}\Delta X)^{-1}>0$. Furthermore, thanks to Yor's formula and the continuity of $A$, we conclude that for any bounded $\varphi \in {\cal L}(X,\mathbb H)$, 
\begin{eqnarray*}
{\cal E}(\varphi\cdot X){\widetilde Z}={\cal E}\left(\varphi\cdot X+[\varphi\cdot X, K]+K-{\widetilde V}\right).\end{eqnarray*}
It is easy to check that (\ref{processV}) and (\ref{C6forX}) imply that $\varphi\cdot X+[\varphi\cdot X, K]$ is a special semimartingale and its compensator $\left(\varphi\cdot X+[\varphi\cdot X, K]\right)^{p,\mathbb H}$ is dominated by ${\widetilde V}$. This proves that the process $\varphi\cdot X+[\varphi\cdot X, K]+K-{\widetilde V}$ is a local supermartingale. As a consequence, ${\cal E}(\varphi\cdot X){\widetilde Z}$ is a positive supermatingale, and hence ${\widetilde Z}\in {\cal D}(X,\mathbb H)$ on the one hand. On the other hand, due to It\^o, we derive 
$$-\ln({\widetilde Z})=\mbox{local martingale}+ {\widetilde V}+{1\over{2}}\widetilde\varphi^{tr}c\widetilde\varphi\cdot A+ \left[-{{\widetilde\varphi^{tr} x }\over{1 +\widetilde\varphi^{tr} x }}+\ln(1 +\widetilde\varphi^{tr} x)\right]\star\mu.$$ By combining this with (\ref{Condi11}), we deduce that  ${\widetilde Z}\in {\cal D}_{log}(X,\mathbb H)$. This proves that assertion (a) holds. Furthermore, since ${\widetilde Z}$ is a positive supermartingale, ${\widetilde Z}^{-1}$ is a positive semimartingale, and 
$${\widetilde Z}_{-}I_{\{\vert \widetilde\varphi\vert\leq n\}}\cdot ({\widetilde Z})^{-1}=\widetilde\varphi I_{\{\vert \widetilde\varphi\vert\leq n\}}\cdot X.$$ Since the LHS term, of the above equality converges (in probability at any time $t\in (0,T]$), we deduce that  $\widetilde\varphi\in L(X,\mathbb H)$ (i.e. it is $X$-integrable in the semimartingale sense), and 
$({\widetilde Z})^{-1}={\cal E}(\widetilde\varphi\cdot X).$
Therefore, on the one hand, this ends the proof for the properties (\ref{integrabilities})-(\ref{duality}). On the other hand, we notice that $\ln({\cal E}(\widetilde\varphi\cdot X)_T)=-\ln(\widetilde Z_T)$ is an integrable random variable, and  for any $\varphi \in {\cal L}(X,\mathbb H)\cap L(X,\mathbb H)$ satisfying the condition $E\ln^-({\cal E}(\varphi\cdot X)_T)<+\infty$, we get
$$
E[\ln({\cal E}(\varphi\cdot X)_T/{\cal E}(\widetilde\varphi\cdot X)_T)]=E[\ln({\cal E}(\varphi\cdot X)_T)- E\ln({\cal E}(\widetilde\varphi\cdot X)_T)]\leq 0.$$ Thus, assertion (d) and (\ref{equalityFiTheta}) follow, and the rest of this step proves assertion (c).\\
Let $Z\in {\cal D}_{log}(X,\mathbb H)$, and by applying Theorem \ref{DlogSet}, we deduce the existence of $(\beta, f, V)$ such that  
\begin{eqnarray*}
&&\varphi^{tr}x f(x)\geq -[f(x)-1-\ln(f(x))]+\ln(1+\varphi^{tr}x ),\quad\mbox{ for any}\ \varphi\in {\cal L}(X,\mathbb H),\\
&&V \succeq \left(\varphi^{tr}b +\varphi^{tr}c\beta+\int \left(\varphi^{tr}xf(x)-\varphi^{tr}h(x)\right)F(dx)\right)\cdot A,\\
&&E[-\ln(Z_T)]\geq E\left[V_T+{1\over{2}}\beta^{tr}c\beta\cdot A_T+\int [f(x)-1-\ln(f(x))]F(dx)\cdot A_T\right].
\end{eqnarray*}
Then by combining these properties (take $\varphi=\widetilde\varphi$) with (\ref{Condi11})-(\ref{processV})-(\ref{C6forX}) and the fact that under (\ref{C6forX}) we have $ \widetilde\varphi^{tr}(b-c\widetilde\varphi)+\displaystyle\int [{{\widetilde\varphi^{tr}x}\over{1+\widetilde\varphi^{tr}x}}-\widetilde\varphi^{tr}h(x)] F(dx)\geq 0$, we derive 
\begin{eqnarray*}
E[-\ln(\widetilde Z_T)]&&=E\left[ {\widetilde V}_T+{1\over{2}}(\widetilde\varphi ^{tr}c\widetilde\varphi \cdot A)_T+(-{{\widetilde\varphi^{tr}x}\over{1+\widetilde\varphi^{tr}x}}+\ln(1+(\widetilde\varphi^{tr}x) )\star\mu)_T\right]\\
&&=E\left[ {\widetilde V}_T+\left({1\over{2}}\widetilde\varphi ^{tr}c\widetilde\varphi+\int(-{{\widetilde\varphi^{tr}x}\over{1+\widetilde\varphi^{tr}x}}+\ln(1+(\widetilde\varphi^{tr}x) )F(dx)\right)\cdot A_T\right]\\
&&=E\left[(\widetilde\varphi^{tr}b-{1\over{2}}\widetilde\varphi^{tr}c\widetilde\varphi)\cdot A_T+\int(\ln(1+\widetilde\varphi^{tr}x)-\widetilde\varphi^{tr}h(x))F(dx)\cdot A_T\right]\\
&&\leq E\left[(\widetilde\varphi^{tr}b-{1\over{2}}\widetilde\varphi^{tr}c\widetilde\varphi)\cdot A_T+\int(\widetilde\varphi^{tr}xf(x)-\widetilde\varphi^{tr}h(x))F(dx)\cdot A_T\right]+\\
&&\hskip 1.5cm+E\left[\left(\int [f(x)-1-\ln(f(x))]F(dx)\right)\cdot A_T\right]\\
&&\leq E\left[\left(-\widetilde\varphi^{tr}c\beta-{1\over{2}}\widetilde\varphi^{tr}c\widetilde\varphi+\int [f(x)-1-\ln(f(x))]F(dx)\right)\cdot A_T+V_T\right]\\
&&\leq E\left[V_T+{1\over{2}}\beta^{tr}c\beta\cdot A_T+\int [f(x)-1-\ln(f(x))]F(dx)\cdot A_T\right]\\
&&\leq E[-\ln(Z_T)].
\end{eqnarray*}
This proves assertion (c), and the first step is complete.\\
%%%%%%%%%%%%%%%%%%%%%%%%%%%%%%%%%%STEP2
{\bf Step 2.} This step proves (d)$\Longrightarrow$ (a). Thus, we suppose that assertion (d) holds. Then there exists  a portfolio ${\widetilde\theta}\in {\Theta(X,\mathbb H)}$ such that  (\ref{PrimalSolution}) holds. Thanks to \cite[Theorem 2.8]{ChoulliDengMa} (see also \cite{ChristensenLarsen2007}  and \cite[Theorem 2.3]{HulleySchweizer}), we deduce that  ${\cal D}(X,\mathbb H)\not=\emptyset$. By combining this with $1+({\widetilde\theta}\cdot X)_T>0$, we conclude the positivity of both processes $1+{\widetilde\theta}\cdot X$ and $1+({\widetilde\theta}\cdot X)_-$, and hence the existence of ${\widetilde\varphi}\in {\cal L}(X,\mathbb  H)\cap L(X, \mathbb H)$ such that  $1+{\widetilde\theta}\cdot X= {\cal E} ({\widetilde\varphi}\cdot X)$ on the one hand. On the other hand, the condition  ${\cal D}(X,\mathbb H)\not=\emptyset$ is equivalent to the existence of the num\'eraire portfolio, that we denote by $\widehat\varphi$ (see \cite{ChoulliDengMa,KabanovKardarasSong,KardarasKaratzas} and the references therein to cite few). This means that there exists $\widehat\varphi\in L(X,\mathbb H)$ such that ${\cal E}(\widehat\varphi\cdot X)>0$ and ${\cal E}(\varphi\cdot X)/{\cal E}(\widehat\varphi\cdot X)$ is a supermartingale for any $\varphi\in L(X,\mathbb H)$ with $1+\varphi\Delta X\geq 0$. In particular, the process 
$$M:= {{{\cal E}(\widetilde\varphi\cdot X)}\over{{\cal E}(\widehat\varphi\cdot X)}}-1,$$
is a suprermartingale. Due to $\ln(x)\leq x-1$, we get  
\begin{eqnarray}\label{inequality0}
-\ln({\cal E}(\widehat\varphi\cdot X))\leq -\ln({\cal E}(\widetilde\varphi\cdot X))+{{{\cal E}(\widetilde\varphi\cdot X)}\over{{\cal E}(\widehat\varphi\cdot X)}}-1,\end{eqnarray}
and deduce that $\ln^-({\cal E}_T(\widehat\varphi\cdot X))$ is integrable. As a result, $\widehat\theta:=\widehat\varphi{\cal E}(\widehat\varphi\cdot X)_{-}\in \Theta(X,\mathbb H)$, and the following hold
\begin{eqnarray}\label{inequality1}
E[\ln({\cal E}_T(\widehat\varphi\cdot X))]&&=E[\ln (1+(\widehat\theta\cdot X)_T)]\nonumber\\
&&\leq E[\ln (1+(\widetilde\theta\cdot X)_T)]=E[\ln({\cal E}_T(\widetilde\varphi\cdot X))].\end{eqnarray}
This, in particular, implies that $\ln({\cal E}_T(\widehat\varphi\cdot X))$ is an integrable random variable, or equivalently $\ln({\cal E}_T(\widetilde\varphi\cdot X)/{\cal E}_T(\widehat\varphi\cdot X)=\ln({\cal E}_T(\widetilde\varphi\cdot X))-\ln({\cal E}_T(\widehat\varphi\cdot X)$ is integrable. Then using Jensen's inequality, we deduce that $$E[\ln({\cal E}_T(\widetilde\varphi\cdot X)/{\cal E}_T(\widehat\varphi\cdot X)]\leq \ln(E[{\cal E}_T(\widetilde\varphi\cdot X)/{\cal E}_T(\widehat\varphi\cdot X)])\leq 0.$$This combined with (\ref{inequality1}) implies that 
\begin{eqnarray}\label{equality1}
E[\ln({\cal E}_T(\widehat\varphi\cdot X))]= E[\ln({\cal E}_T(\widetilde\varphi\cdot X))].\end{eqnarray}
A combination of this with (\ref{inequality0}) leads to $E[{{{\cal E}_T(\widetilde\varphi\cdot X)/{\cal E}_T(\widehat\varphi\cdot X)}}]=1$, and hence the process $M+1$ is in fact a martingale (a positive supermartingale with constant expectation  is a martingale). It is clear that $f(x):=x-\ln(1+x),\ x>-1$, is a nonnegative and strictly convex function that vanishes at $x=0$ only.  Since $E[f(M_T)]<+\infty$, we conclude that $f(M)$ is a nonnegative submartingale satisfying
$$0=E[f(M_0)]\leq E[f(M_t)]\leq  E[f(M_T)]=0,$$
where the last equality follows from combining  (\ref{equality1}) with the fact that $M$ is martingale. Thus, we conclude that $M\equiv 0$ and hence ${\cal E}(\widehat\varphi\cdot X)\equiv {\cal E}(\widetilde\varphi\cdot X)$. As a consequence the process $Z:=1/ {\cal E}(\widetilde\varphi\cdot X)$ belongs to ${\cal D}(X,\mathbb H)$. Therefore, assertion (a) follows immediately from this and 
$$E[-\ln({ Z}_T)]=E[\ln({\cal E}_T ({\widetilde\varphi}\cdot X))]=E[1+(\widetilde\theta\cdot X)_T]<+\infty,$$
and the proof of (d)$\Longrightarrow$ (a) is complete.\\
{\bf Step 3.} This step proves the implication (a) $\Longrightarrow$ (b). Hence, we assume that assertion (a) holds for the rest of this proof. In virtue of Theorem \ref{DlogSet}, which guarantees the existence of $(\beta, f, V)$ such that $\beta\in L(X^c,\mathbb H)$, $f$ is $\widetilde{\cal P}(\mathbb H)$-measurable, positive and $\sqrt{(f-1)^2\star\mu}\in {\cal A}^+_{loc}$,  $V$ is a predictable and nondecreasing process, and  the following hold for any bounded $\theta\in {\cal L}(X,\mathbb H)$. 
\begin{eqnarray}
&&E\left[V_T+{1\over{2}}(\beta^{tr}c\beta\cdot A)_T+\left(\int(f(x)-1-\ln(f(x)))F(dx)\right)\cdot A_T\right]\nonumber\\
&&\leq E[-\ln( Z_T)]<+\infty,\label{C0}\hskip 1cm\\
%&&+\infty>\geq {\cal E}(N-V),\ N:=\beta\cdot X^c+(f-1)\star(\mu-\nu),\label{Cond0}\\
&&\left(\int \vert f(x)\theta^{tr}x-\theta^{tr}h(x)\vert F(dx)\right)\cdot A_T<+\infty\  P\mbox{-a.s.},\ \mbox{ and}\label{C1}\\
&&\left(\theta^{tr}b+\theta^{tr}c\beta+\int [f(x)\theta^{tr}x-\theta^{tr}h(x)] F(dx)\right)\cdot A\preceq V,\label{C2}
\end{eqnarray}
The rest of this proof  is divided into two sub-steps, and uses these properties. The first sub-step proves that a functional $L$, that we will define below, attains its minimal value, while the second sub-step proves that this minimum fulfills (\ref{Condi11})-(\ref{processV})-(\ref{C6forX}).\\
{\bf Step 3.a.}  Throughout the rest of the proof,  we denote by $L_{(\omega,t)}$ --$P\otimes A$-almost all $(\omega, t)\in \Omega\times [0,+\infty)$-- the function given by 
 \begin{eqnarray}\label{FunctionL}
L_{(\omega,t)}(\lambda):=-\lambda^{tr}b(\omega,t)+{1\over{2}}\lambda^{tr}c(\omega,t)\lambda +\int \left(\lambda^{tr}h(x)-\ln((1+\lambda^{tr}x)^+)\right)F_{(\omega,t)}(dx),\hskip 0.5cm
\end{eqnarray}
for any $\lambda\in{\mathbb R}^d$ with the convention $\ln(0^+)=-\infty$. This sub-step proves the existence of a predictable process $\widetilde{\varphi}$ such that $P\otimes A$-almost all $ (\omega,t)\in\Omega\times[0,+\infty)$
\begin{eqnarray}\label{minimizationpb}
\widetilde{\varphi}(\omega,t)\in {\cal L}_{(\omega,t)}(X,\mathbb H)\quad\rm{and}\quad L_{(\omega,t)}(\widetilde{\varphi}(\omega,t))=\min_{\lambda\in{\cal L}_{(\omega,t)}(X,\mathbb H)}  L_{(\omega,t)}(\lambda).\end{eqnarray}
To this end, we start by noticing that in virtue of a combination of Remark \ref{remark3.3} (which implies that this functional takes values in $(-\infty,+\infty]$), Lemma \ref{measurabilitylemma}, and \cite[Proposition 1]{Evstigneev} (which guarantees the existence of a predictable selection for the minimizer when it exists), this proof boils down to prove that $L_{(\omega, t)}$ attains in fact its minimum for all $(\omega, t)\in\Omega\times[0,+\infty)$. This is the aim of the rest of this sub-step. For the sake of simplicity, we denote $L_{(\omega,t)}(\cdot)$ by $L$ throughout the rest of this proof. In order to prove that $L$ attains it minimum value, we start by  proving that this function $L$ is convex, proper and closed. Let first recall some definitions from convex analysis. Consider a convex function $f$. The effective domain of $f$, denoted by dom($f$), is the set of all $x\in \mathbb R^d$  such that $f(x)<+\infty$. The function $f$ is said to be proper if, for any $x\in\mathbb R^d$, $f(x)>-\infty$ and if its effective domain dom($f$ ) is not empty. For all undefined or unexplained concepts from convex analysis, we refer the reader to Rockafellar \cite{Rockafellar}. \\
Let $\theta$  be a bounded element of ${\cal L}(X,\mathbb H)$, and due to  $\ln(1+(\theta^{tr}x)^+)\leq (\theta^{tr}x)^+ \leq \vert\theta\vert\vert x\vert$ and $\int_{(\vert x\vert>1)}\vert x
\vert F(dx)<+\infty$ (since $X$ is $\sigma$-special), we obtain $P\otimes A$-a.e.
\begin{eqnarray}\label{C1forX}
\int_{(\vert x\vert>1)}\ln(1+(\theta^{tr}x)^+) F(dx)\leq \int_{(\vert x\vert>1)} (\theta^{tr}x)^+ F(dx)\leq \vert\theta\vert\int_{(\vert x\vert>1)} \vert x\vert F(dx)<+\infty.
\end{eqnarray}
Then by combining this with 
\begin{eqnarray*}
\displaystyle\int \left(\lambda^{tr}h(x)-\ln(1+\lambda^{tr}x)\right)F(dx)\geq -\int_{(\vert x\vert> 1)}\ln(1+(\lambda^{tr}x)^{+})F(dx)>-\infty,\end{eqnarray*}
and $L(0)=0<+\infty$ (i.e. $0\in\mbox{dom}(L)\subset {\cal L}(X,\mathbb H)$), we deduce that $L$ is a convex and proper function. Now we prove that $L$ is closed or equivalently $L$ is lower semi-continuous. Let $\theta_n$ be a sequence in $\mathbb R^d$ that converges to $\theta$ such that $L(\theta_n)$ converges. Then it is clear that 
$\theta_n^{tr}b+\theta_n^{tr}c\beta$  converges to $\theta^{tr}b+\theta^{tr}c\beta$ and  $\int_{(\vert x\vert>1)}\ln(1+(\theta_n^{tr}x)^+) F(dx)$  converges to $\int_{(\vert x\vert>1)}\ln(1+(\theta^{tr}x)^+) F(dx)$. This latter is due to a combination of $\ln(1+(\theta_n^{tr}x)^+)\leq (\theta_n^{tr}x)^+\leq (\sup_n\vert\theta_n\vert) \vert x\vert$, (\ref{C1forX}), and dominated convergence theorem. Now consider the assumption 
\begin{equation}\label{condition4theta}
\theta\in{\cal L}(X,\mathbb H)\quad\mbox{and there exists $n_0$ such that for all}\  n\geq n_0\quad \theta_n\in{\cal L}(X,\mathbb H).
\end{equation}
Under (\ref{condition4theta}), by combining Fatou's lemma and the above remarks, we get
\begin{eqnarray*}
L(\theta)&&=-\theta^{tr}b+{1\over{2}}\theta^{tr}c\theta +\int \left(\theta^{tr}h(x)-\ln(1+\theta^{tr}x)\right)F(dx)\\
&&=-\theta^{tr}b+{1\over{2}}\theta^{tr}c\theta-\int_{\vert x\vert>1}\ln(1+(\theta^{tr}x)^+)F(dx)\\
&&-\int_{\vert x\vert>1}\ln(1-(\theta^{tr}x)^-)F(dx)+
\int_{\vert x\vert\leq 1}(\theta^{tr}x- \ln(1+\theta^{tr}x))F(dx)\\
&&\leq \lim_{n\longrightarrow+\infty}L(\theta_n).
\end{eqnarray*}
This proves that $L$ is closed under (\ref{condition4theta}) on the one hand. On the other hand, it is clear that, when  (\ref{condition4theta}) is violated, there exists a subsequence $(\theta_{k(n)})_n$ such that $\theta_{k(n)}\not\in{\cal L}(X,\mathbb H)$ for all $n\geq 1$. As a result, since $L(\theta_n)$ converges, we conclude that  $L(\theta)\leq \lim_{n\longrightarrow+\infty}L(\theta_n)= \lim_{n\longrightarrow+\infty}L(\theta_{k(n)})=+\infty$. 
This proves that $L$ is closed, convex and proper. Thus, we can apply \cite [Theorem 27.1(b)]{Rockafellar} which states that, for $L$ to attain its minimal value, it is sufficient to prove that the set of recession for $L$ is contained in the set of directions in which $L$ is constant. To check this last condition, we calculate the recession function for $L$. For $\lambda\in\mbox{dom}(L )$ and $y\in\mathbb R^d$, the recession function for $L$ is by definition
$$L0^+(y):=\lim_{\alpha\longrightarrow+\infty} {{L(\lambda+\alpha y)-L(\lambda)}\over{\alpha}}.$$
Consider the following sets
$$\Gamma^+(\lambda):=\{x\in\mathbb R^d\quad \big|\quad \lambda^{tr}x> 0\},\quad \Gamma^-(\lambda):=\{x\in\mathbb R^d\quad \big|\quad \lambda^{tr}x< 0\},$$
and remark that we have 
\begin{eqnarray*}
&&{{L(\lambda+\alpha y)-L(\lambda)}\over{\alpha}}\\
&&=-y^{tr}b+{{\alpha }\over{2}}y^{tr}cy+y^{tr}c\lambda+\int \left(y^{tr}h(x)-{1\over{\alpha}}\ln(1+{{\alpha y^{tr}x}\over{1+\lambda^{tr}x}})\right)F(dx)\\
&&=-y^{tr}b+{{\alpha }\over{2}}y^{tr}cy+y^{tr}c\lambda+\int_{\Gamma^+(y)} \left(y^{tr}h(x)-{1\over{\alpha}}\ln(1+{{\alpha y^{tr}x}\over{1+\lambda^{tr}x}})\right)F(dx)\\
&&+\int_{\Gamma^-(y)} \left(y^{tr}h(x)-{1\over{\alpha}}\ln(1+{{\alpha y^{tr}x}\over{1+\lambda^{tr}x}})\right)F(dx).
\end{eqnarray*}
Then, on the one hand, we calculate the recession function $L0^+(y)$ as follows.
\begin{eqnarray*}L0^+(y)=\left\{\begin{array}{lll} +\infty\hskip 4cm \mbox{if either}\ F(\Gamma^-(y))>0 \ \mbox{or}\ y^{tr}cy>0,\\
-y^{tr}b+\int_{\Gamma^+(y)} y^{tr}h(x)F(dx)\hskip 0.30cm  \mbox{otherwise}\end{array}\right.\end{eqnarray*}
On the other hand, we have 
\begin{eqnarray*}
&&\alpha \{y\in\mathbb R^d\ :\ cy=0\ \mbox{and}\ F(\Gamma^-(y))=0\}\subset {\cal L}(X,\mathbb H)\ \mbox{ for any}\ \alpha\in(0,+\infty),\\
&&-\alpha \{y\in\mathbb R^d\ :\ cy=0\ \mbox{and}\ F(\Gamma^+(y))=0\}\subset {\cal L}(X,\mathbb H)\ \mbox{ for any}\ \alpha\in(0,+\infty).
\end{eqnarray*}
 Thus, by combining these with (\ref{C2}), we deduce that 
\begin{eqnarray*}
&&-y^{tr}b+\int_{\Gamma^+(y)} y^{tr}h(x)F(dx)> 0\quad \mbox{if}\ F(\Gamma^-(y))=0<F(\Gamma^+(y)),\ cy=0,\\
&&y^{tr}b-\int_{\Gamma^-(y)} y^{tr}h(x)F(dx)> 0\quad \mbox{if}\ F(\Gamma^+(y))=0<F(\Gamma^-(y)),\ cy=0.
\end{eqnarray*}
Thus, thanks to these remarks, the recession cone for $L$ and the set of directions in which $L$ is constant, that we denote RC and CD respectively, are defined and calculated as follows
\begin{eqnarray*}
RC&&:=\{y\in{\mathbb  R}^d\ \big|\ L0^+(y)\leq 0\}=\{y\in\mathbb R^d\ \big|\ cy=y^{tr}b= F(\Gamma^-(y))=F(\Gamma^+(y))=0\}\\
CD&&:=\{y\in{\mathbb  R}^d\ \big|\ L0^+(y)\leq 0,\  L0^+(-y)\leq 0 \}\\
&&=\{y\in\mathbb R^d\ \big|\ y^{tr}b=cy=F(\Gamma^-(y))=F(\Gamma^+(y))=0\}.
\end{eqnarray*}
This proves that both sets (RC and CD) are equal. Hence, thanks to \cite [Theorem 27.1(b)]{Rockafellar}, we conclude that $L_{(\omega,t)}$ attains its minimal value at $\varphi(\omega,t)$  which satisfies  $1+x^{tr}  \varphi(\omega,t)>0$ $F_{(\omega,t)}(dx)$-a.e. since $L_{(\omega,t)}(  \varphi(\omega,t))\leq L_{(\omega,t)}(0)=0<+\infty$. This ends the first part of the third step.  \\
{\bf Step 3.b.} This sub-step proves that $\widetilde\varphi$, a minimizer for $L$ proved in the previous sub-step, fulfills in fact the conditions of assertion (b) (i.e. the properties (\ref{Condi11})-(\ref{processV})-(\ref{C6forX})).  \\
%%%%%%%%%%%%%%%%%%%%%%%%%%%%%%%%%%%%%%%%%%%%%%%%%%%%%%%
%%%%%%%%%%%%%%%%%%%%%%%%%%%%%%%%%%%%%%%%%%%%%%%%%%%%%%%%%%%%%%%%%%%%%%%%%%%%%%%%%%%
Since $L(\widetilde\varphi)\leq L(\varphi)$ for any $\varphi\in{\cal L}(X,\mathbb H)$.  Let  $\varphi\in{\cal L}(X,\mathbb H)$ and $\alpha\in(0,1)$, then using similar calculations as above, we get 
\begin{eqnarray*}
&&{{L(\widetilde\varphi)-L(\widetilde\varphi+\alpha(\varphi-\widetilde\varphi))}\over{\alpha}}=(\varphi-\widetilde\varphi)^{tr}b-{{\alpha }\over{2}}(\varphi-\widetilde\varphi)^{tr}c(\varphi-\widetilde\varphi)-(\varphi-\widetilde\varphi)^{tr}c\widetilde\varphi+\\
&&+\int \left({1\over{\alpha}}\ln(1+{{\alpha (\varphi-\widetilde\varphi)^{tr}x}\over{1+{\widetilde\varphi}^{tr}x}})-(\varphi-\widetilde\varphi)^{tr}h(x)\right)F(dx).
\end{eqnarray*}
It is clear that, as a function of $\alpha$, $\alpha^{-1}\ln(1+{{\alpha (\varphi-\widetilde\varphi)^{tr}x}\over{1+{\widetilde\varphi}^{tr}x}})$ is decreasing and hence
\begin{eqnarray*} \ln(1+ \varphi^{tr}x)- \ln(1+\widetilde\varphi^{tr}x)\leq {1\over{\alpha}}\ln(1+{{\alpha (\varphi-\widetilde\varphi)^{tr}x}\over{1+{\widetilde\varphi}^{tr}x}})\leq {{(\varphi-\widetilde\varphi)^{tr}x}\over{1+{\widetilde\varphi}^{tr}x}}.\end{eqnarray*}
As a result of this, combined with the convergence monotone theorem, we deduce that $$\int \left({1\over{\alpha}}\ln(1+{{\alpha (\varphi-\widetilde\varphi)^{tr}x}\over{1+{\widetilde\varphi}^{tr}x}})-(\varphi-\widetilde\varphi)^{tr}h(x)\right)F(dx)$$ converges to $\int[ {{(\varphi-\widetilde\varphi)^{tr}x}\over{1+{\widetilde\varphi}^{tr}x}}-(\varphi-\widetilde\varphi)^{tr}h(x)]F(dx),$ when $\alpha$ goes to zero  and hence (\ref{C6forX}) is proved. By using (\ref{C6forX})  for $\varphi=0$, and $L(\widetilde\varphi)\leq L(0)=0$, we get 
\begin{eqnarray}
&&0\leq \widetilde\varphi^{tr}b-\widetilde\varphi^{tr}c\widetilde\varphi+ \int \left(-\widetilde\varphi^{tr}h(x)+{{\widetilde\varphi^{tr}x}\over{1+{\widetilde\varphi}^{tr}x}}\right)F(dx)\label{Positive1}\\
&&0\leq \widetilde\varphi^{tr}b-{1\over{2}}\widetilde\varphi^{tr}c\widetilde\varphi+\int \left(\ln(1+\widetilde\varphi^{tr}x)-\widetilde\varphi^{tr}h(x)\right)F(dx).\label{Positive1}
\end{eqnarray}
Therefore, by combining these two inequalities with $-{1\over{2}}\widetilde\varphi^{tr}c\widetilde\varphi-{1\over{2}}\beta^{tr}c\beta \leq  \widetilde\varphi^{tr}c\beta$, and  $f(x)-1-\ln(f(x))\geq \ln(1+\widetilde\varphi^{tr}x)-f(x)\widetilde\varphi^{tr}x$ (Young's inequality), we derive 
\begin{eqnarray*}
&&\widetilde\varphi^{tr}b-{1\over{2}}\widetilde\varphi^{tr}c\widetilde\varphi-{1\over{2}}\beta^{tr}c\beta +\int [\ln(1+\widetilde\varphi^{tr}x)-\widetilde\varphi^{tr}h(x)] F(dx)+\\
&&\hskip 2.5cm -\int [f(x)-1-\ln(f(x))] F(dx) \\
&&\leq\widetilde\varphi^{tr}b-{1\over{2}}\widetilde\varphi^{tr}c\widetilde\varphi-{1\over{2}}\beta^{tr}c\beta +\int [f(x)\widetilde\varphi^{tr}x-\widetilde\varphi^{tr}h(x)] F(dx)\\
&&\leq \widetilde\varphi^{tr}b+\widetilde\varphi^{tr}c\beta+\int [f(x)\widetilde\varphi^{tr}x-\widetilde\varphi^{tr}h(x)] F(dx).
\end{eqnarray*}
Therefore, thanks to this latter inequality and (\ref{C2}),  we deduce that 
\begin{eqnarray*}
0&&\preceq\left( \widetilde\varphi^{tr}b-{1\over{2}}\widetilde\varphi^{tr}c\widetilde\varphi+\int [\ln(1+\widetilde\varphi^{tr}x)-\widetilde\varphi^{tr}h(x)] F(dx)\right)\cdot A\\
&&\preceq\left(\int [f(x)-1-\ln(f(x))] F(dx)+{1\over{2}}\beta^{tr}c\beta\right)\cdot A+V.
\end{eqnarray*}
By combining this, (\ref{C0}),  the fact that 
\begin{eqnarray*}
&&\widetilde\varphi^{tr}b-{1\over{2}}\widetilde\varphi^{tr}c\widetilde\varphi+\int [\ln(1+\widetilde\varphi^{tr}x)-\widetilde\varphi^{tr}h(x)] F(dx)\\
&&=
\left({1\over{2}}\widetilde\varphi^{tr}c\widetilde\varphi+\int \left(\ln(1+\widetilde\varphi^{tr}x)- {{\widetilde\varphi^{tr}x}\over{1+{\widetilde\varphi}^{tr}x}}\right)F(dx)\right)\\
&&\hskip 0.25cm+\left(\widetilde\varphi^{tr}b-\widetilde\varphi^{tr}c\widetilde\varphi+ \int \left(-\widetilde\varphi^{tr}h(x)+ {{\widetilde\varphi^{tr}x}\over{1+{\widetilde\varphi}^{tr}x}}\right)F(dx)\right),
\end{eqnarray*}
where both terms of the RHS are nonnegative, and the second term of this RHS coincides with ${{d\widetilde V}\over{dA}}$, we conclude that 
\begin{eqnarray*}
E\left[{\widetilde V}_T+\left({1\over{2}}\widetilde\varphi^{tr}c\widetilde\varphi+\int \left(\ln(1+\widetilde\varphi^{tr}x)- {{\widetilde\varphi^{tr}x}\over{1+{\widetilde\varphi}^{tr}x}}\right)F(dx)\right)\cdot A_T\right]<+\infty.
\end{eqnarray*}
This proves (\ref{Condi11}), and assertion (b) follows. This ends the proof of the theorem.\end{proof}

%%%%%%%%%%%%%%%%%%%%%%%%%%%%%%%%%%%%%%%%%%%%%%%%%%%%%%%%%%%%%%%%%%%%%%%%%%%%%
%%%%%%%%%%%%%%%%%%%
%%%%%%%%%%%%%%%%%%%%%%%%%%%%%%%%%%%%%%%%%%%%%%%%%%%%%%%%%%%%%%%%%%%%%%%%%%%%%
%%%%%%%%%%%%%%%%%%%%%%%%%%%%%%%%%%%%%%%%%%%%%%%%%%%%%%%%%%%%%%%%%%%%%%%%%%
%%%%%%%%%%%%%%%%%%%%%%%%%%%%%%%%%%%%%%%%%%%%%%%%%%%%%%%%%%%%%%%%%%%%%% 

%%%\newpage

\appendix

%%\centerline{\Large\bf Appendix}

%%%%%%%%%%%%%%%%%%%%%%%%%%%%%%%%%%%%%%%%%%%%%%%%%%%%%%%%%%%%%%%%%%%%%%%%%%%%%%%
%%%%%%%%%%%%%%%%%%%%%%%%%%%%%%%%%%%%%%%%%%%%%%%%%%%%%%%%%%%%%%%%%%%%%%%%%%%%%%%%%

\section{Some useful integrability properties}
The results of this section are new and are general,  not technical at all, and very useful, especially the first lemma and the proposition.

\begin{lemma}\label{H0toH1martingales} Consider $K\in {\cal M}_{0,loc}(\mathbb H)$ with $1+\Delta K>0$. 
If 
\begin{eqnarray}\label{FiniteCondition}
E[\langle K^c\rangle_T+\sum_{0<s\leq T}(\Delta K_s-\ln(1+\Delta K_s))]<+\infty,\end{eqnarray}
then $E[\sqrt{[K,K]_T}]<+\infty$ or equivalently $E[\displaystyle\sup_{0\leq t\leq T}\vert K_t\vert]<+\infty$.
\end{lemma}

\begin{proof}  Let $K\in {\cal M}_{0,loc}(\mathbb H)$ such that $1+\Delta K>0$ and (\ref{FiniteCondition}) holds. Then it is enough to remark that for $\delta\in (0,1)$, we have 
$$\Delta K-\ln(1+\Delta K)\geq {{\delta\vert \Delta K\vert }\over{\max(2(1-\delta),1+\delta^2)}}I_{\{\vert \Delta K\vert >\delta\}}+{{(\Delta K)^2}\over{1+\delta}}  I_{\{\vert \Delta K\vert \leq \delta\}}.$$
By using this inequality and  (\ref{FiniteCondition}), on the one hand, we deduce that
\begin{eqnarray*}
&&E\left[\langle K^c\rangle_T+\sum_{0<t\leq T}\vert \Delta K_t\vert  I_{\{\vert \Delta K_t\vert >\delta\}}+\sum_{0<t\leq T}(\Delta K)^2  I_{\{\vert \Delta K\vert \leq \delta\}}\right]\\
&&\leq C_{\delta} E\left[\langle K^c\rangle_T+\sum_{0<s\leq T}(\Delta K_s-\ln(1+\Delta K_s))\right] < +\infty,\end{eqnarray*}
where $C_{\delta}:=1+\delta +\max(2(1-\delta),1+\delta^2)/\delta$.  On the other hand, it is clear that 
$$[K,K]^{1/2}_T\leq \sqrt{\langle K^c\rangle}+\sum_{0<t\leq T}\vert \Delta K_t\vert  I_{\{\vert \Delta K_t\vert >\delta\}}+\sqrt{\sum_{0<t\leq T}(\Delta K)^2  I_{\{\vert \Delta K\vert \leq \delta\}}}.$$
This ends the proof of the lemma.
\end{proof}
%%%%%%%%%%%%%%%%%%%%%%%%%%%%%%%%%%%%%%%%%%%%%%%%%%%%%%%%%%%%%%%%%%%%%%%%%%%%
%%%%%%%%%%%%%%%%%%%%%%%%%%%%%%%%%%%%%%%%%%%%%%%%%%
%%%%%%%%%%%%%%%%%%%%%%%%%%%%%%%%%%%%%%%%%%%%%%%%%%%%%%%%%%%%%%%%%%%%%%%%%%
%%%%%%%%%%%%%%%%%%%%%%%%%%%%%%%%%%%%%%%%%%%%%%%%%%%%%%%%%%%%%%%%%%%%%%%%%%%
\begin{lemma}\label{lemm4F-Gintregrability}
Let $\lambda\in{\cal L}(X,\mathbb H)$, and $\delta\in(0,1)$ such that 
\begin{eqnarray}
{{\vert\lambda^{tr}x\vert}\over{1+ \lambda^{tr}x}} I_{\{\vert\lambda^{tr}x\vert>\delta\}}\star\mu+\left({{\lambda^{tr}x}\over{1+ \lambda^{tr}x}}\right)^2 I_{\{\vert\lambda^{tr}x\vert\leq\delta\}}\star\mu\in{\cal A}^+_{loc}(\mathbb H).\end{eqnarray}
Then $\sqrt{((1+ \lambda^{tr}x)^{-1}-1)^2\star\mu}\in {\cal A}^+_{loc}(\mathbb H)$.
\end{lemma}
\begin{proof}
By using $\sqrt{\sum_i x_i^2}\leq \sum_i \vert x_i\vert$, we derive 
\begin{eqnarray*}
&&\sqrt{((1+ \lambda^{tr}x)^{-1}-1)^2\star\mu}=\sqrt{\sum \left({{ \lambda^{tr}\Delta X}\over{1+\lambda^{tr}\Delta X}}\right)^2}\\
&&\leq \sqrt{\sum {{ (\lambda^{tr}\Delta X)^2}\over{(1+\lambda^{tr}\Delta X)^2}}I_{\{\vert  \lambda^{tr}\Delta X\vert\leq \delta\}}}+\sum {{ \vert\lambda^{tr}\Delta X\vert}\over{1+\lambda^{tr}\Delta X}}I_{\{\vert  \lambda^{tr}\Delta X\vert> \delta\}}.
\end{eqnarray*}
Thus, the lemma follows immediately from the latter inequality.
\end{proof}
%%%%%%%%%%%%%%%%%%%%%%%%%%%%%%%%%%%%%%%%%%%%%%%%%%%%%%%%%%%%%%%%%%%%
%%%%%%%%%%%%%%%%%%%%%%%%%%%%%%%%%%%%%%%%%%%%%%%%%%%%%%%%%%%%%%%%%%%
\section{Martingales and deflators via predictable characteristics}
For the  following representation theorem, we refer to
\cite[Theorem 3.75]{J79}  and to \cite[Lemma 4.24]{JS03}.

\begin{theorem}\label{tmgviacharacteristics} Suppose that $X$ is quasi-left-continuous, and let $N\in {\cal M}_{0,loc}(\mathbb H)$. Then, there exist $\phi\in L(X^c,\mathbb H)$, $N'\in {\cal M}_{0,loc}(\mathbb H)$ with
$[N',X]=0$ and functionals $f\in {\widetilde{{\cal P}}}(\mathbb H)$ and $g\in
{\widetilde{{\cal O}}}(\mathbb F)$ such that the following hold.\\
{\rm{(a)}} $
 \Bigl (\displaystyle\sum_{s=0}^t (f(s, \Delta S_s )-1)^2
I_{\{\Delta S_s\not = 0\}}\Bigr )^{1/2}$ and $\Bigl (\displaystyle\sum_{s=0}^t g(s, \Delta S_s )^2
I_{\{\Delta S_s\not = 0\}}\Bigr )^{1/2}$ belong to ${\cal A}^+_{loc}$.\\
{\rm{(b)}}  $
M^P_{\mu}(g\ |\ {\widetilde {{\cal P}}})=0,$ $P\otimes\mu$-.a.e., and the process $N$ is given by \begin{equation}
\label{Ndecomposition}
 N=\phi\cdot X^c+(f-1)\star(\mu-\nu)+g\star\mu+{N'}.
%\int f(x)\nu(\{t\},dx),
\end{equation} \end{theorem}
%%%%%%%%%%%%%%%%%%%%%%%%%%%%%%%%%%%%%%%%%%%%%%%%%%%%%%%%%%%%%%%%
%%%%The quadruplet $ (\beta, f, g, N') $ is called  by Jacod's components of $N$ (under $P$).\\

%%%%%%%%%%%%%%%%%%%%%%%%%%%%%%%%%%%%%%%%%%%%%%%%%%%%%%%%%%%%%%%%%%%%%%%%%%
The following theorem describes how general deflators can be characterized using the predictable characteristics. A version of this theorem can be found in \cite{Ma}.
%%\begin{theorem}\label{theosigmadensityiff}
%%Consider the model $\left(X,\mathbb H,Q\right)$ is defined above with its predictable characteristics $(b,c,F,A)$.
%%Then, $\left(X,\mathbb H,Q\right)$ satisfies NUPBR if and only if there exists a pair $(\beta, f)$ of $\mathbb H$-predictable process and $\widetilde{\cal P}(\mathbb H)$-measurable functional satisfying the following
 %% \begin{eqnarray}
 %% &&f>0,\ Q\times\mu_X-a.e.,\ \ \ \beta^{tr}c\beta\cdot A+\sqrt{(f-1)^2\star\mu_X}\in {\cal A}^+_{loc}(\mathbb H,Q),\hskip 1cm \label{integrabilitycondition0}\\
  %% &&\displaystyle\int \vert xf(x)-h(x)\vert F(dx)<+\infty,\ \ \ Q\times
 %%  A-a.e.\ \ \ \mbox{and}\label{integrabilitycondition1}\\
  %% &&b + c\beta + \displaystyle\int\Bigl(x f(x)- h(x)\Bigr)F(dx) = 0,\ \ \ Q\times
  %% A-a.e.\label{integrabilitycondition2}\end{eqnarray}
 %% \end{theorem}
%% \begin{proof}
 %% The proof follows  \cite{AkasmitChoulliDengJeanblanc1} and \cite[Lemma 2.4]{choullistricker07}.
%% \end{proof}
\begin{theorem}\label{DlogSet} Suppose $X$ is quasi-left-continuous. 
$Z\in {\cal D}_{log}(X,\mathbb H)$ if and only if  there exists a triplet  $(\beta, f, V)$ such that $\beta\in L(X^c,\mathbb H)$, $f$ is $\widetilde{\cal P}(\mathbb H)$-measurable, positive and $\sqrt{(f-1)^2\star\mu}$ belongs to ${\cal A}^+_{loc}(\mathbb H)$,  $V$ is an $\mathbb H$-predictable and nondecreasing process, and  the following hold for any bounded process $\theta\in {\cal L}(X,\mathbb H)$.
\begin{eqnarray}
&&Z={\cal E}\Bigl(\beta\cdot X^c+(f-1)\star(\mu-\nu)\Bigr)\exp(-V),\label{Cond0}\\
&&E\left[V_T+\left({1\over{2}}\beta^{tr}c\beta+\int(f(x)-1-\ln(f(x)))F(dx)\right)\cdot A_T\right]\leq E[-\ln( Z_T)],\hskip 0.8cm\\
%&&+\infty>\geq {\cal E}(N-V),\ N:=\beta\cdot X^c+(f-1)\star(\mu-\nu),\label{Cond0}\\
&&\left(\int \vert f(x)\theta^{tr}x-\theta^{tr}h(x)\vert F(dx)\right)\cdot A_T<+\infty\quad P\mbox{-a.s.}\label{Cond1}\\
&&\left(\theta^{tr}b+\theta^{tr}c\beta+\int [f(x)\theta^{tr}x-\theta^{tr}h(x)] F(dx)\right)\cdot A\preceq V,\label{Cond2}
\end{eqnarray}
\end{theorem}

\begin{proof}
 Let $Z\in {\cal D}_{log}(X,\mathbb H)$, then $Z_{-}^{-1}\cdot Z$ a local supermartingale (which follows from $Z\in {\cal D}(X,\mathbb H)$ only). Hence, there exists a local martingale $N$ and a nondecreasing and predictable process $V$ such that $Z={\cal E}(N)\exp(-V)$. Then we derive
$$-\ln(Z)=-N+V+{1\over{2}}\langle N^c\rangle+\sum (\Delta N-\ln(1+\Delta N)).$$
Thus $Z\in {\cal D}_{log}(X,\mathbb H)$ if and only if $V+{1\over{2}}\langle N^c\rangle+\sum (\Delta N-\ln(1+\Delta N))$ is integrable. Then there exists a positive and ${\widetilde {\cal P}}(\mathbb H)$-measurable functional $f$ such that $\sqrt{(f-1)^2\star\mu}$ is locally integrable, and $\beta\in L(X^c,\mathbb H)$ such that $N$ can be chosen to be $N:=\beta\cdot X^c+(f-1)\star(\mu-\nu)$ and $V=v\cdot A$. Then  $Z\in {\cal D}_{log}(X,\mathbb H)$  if and only if $V+{1\over{2}}\beta^{tr}c\beta\cdot A+ (f-1-\ln(f))\star\nu\in  {\cal A}^+(\mathbb H)$ and $Z{\cal E}(\theta\cdot X)$ is a supermartingale, for any locally bounded $\mathbb H$-predictable process $\theta$ such that $1+\theta^{tr}x>0$ $P\otimes A$-a.e.. Here $(b, c, \nu:=F\otimes A)$ is the predictable characteristics of $(X,\mathbb H)$. \\
On the one hand, we have $Z{\cal E}(\theta\cdot X)={\cal E}(N-v\cdot A+\theta\cdot X+[\theta\cdot X, N])$ is a positive supermartingale and hence $N-v\cdot A+\theta\cdot X+[\theta\cdot X, N]$ is a local supermartingale. This is equivalent, (after simplification and transformation), to the conditions (\ref{Cond1})-(\ref{Cond2}). This ends the proof of theorem. \end{proof}

\section{A measurability result}
%%%%%%%%%%%%%%%%%%%%%%%%%%%%%%%%%%%%%%%%%%%%%%%%%%%%%%%
\begin{lemma}\label{measurabilitylemma}  Consider the triplet $(\Omega\times[0,+\infty),{\cal P}(\mathbb H), P\otimes A)$, and $L(\omega,t, \lambda):=L_{(\omega,t)}(\lambda)$, defined in (\ref{FunctionL}) for any $\lambda\in {\mathbb R}^d$ and any $(\omega,t)\in\Omega\times[0,+\infty)$. Then the functional $L$, as map $(\omega,t,\lambda)\longrightarrow L(
\omega,t,\lambda)$, is ${\cal P}(\mathbb H)\times{\cal B}({\mathbb R}^d)$-measurable.
\end{lemma}
\begin{proof}  The proof of the lemma will be achieved in two steps. The first step defines a family of functionals $\{L_{\delta}(\omega,t,\cdot),\ \delta\in (0,1)\}$ for $(\omega,t)\in\Omega\times[0,+\infty)$, and proves that these functionals are indeed ${\cal P}(\mathbb H)\times{\cal B}({\mathbb R}^d)$-measurable (i.e. jointly measurable in $(\omega,t)$ and $\lambda$). Then the second step proves that $L_{\delta}(\omega, t,\lambda)$ converges to $L(\omega,t,\lambda)$ when $\delta$ goes to one for any $(\omega,t,\lambda)\in\Omega\times[0,+\infty)\times{\mathbb R}^d$.\\
{\bf Step 1:} Let $(\omega,t)\in\Omega\times[0,+\infty)$ and $\delta\in(0,1)$. Then for  all $\lambda\in{\mathbb R}^d$, put
\begin{eqnarray*}
&&L_{\delta}(\omega, t,\lambda):=-\lambda^{tr}b(\omega,t)+{1\over{2}}\lambda^{tr}c(\omega,t)\lambda +\int_{\mathbb R^d}f_{\delta}(\lambda,x)F_{(\omega,t)}(dx),\\
&& f_{\delta}(\lambda,x):=\delta\lambda^{tr}h(x)-\ln(1-\delta+\delta(1+\lambda^{tr}x)^+).\end{eqnarray*}
It is clear that for any $\lambda\in{\mathbb R}^d$, $L_{\delta}(\omega, t,\lambda)$ is predictable. Thus, in order to prove that $L_{\delta}$ is jointly measurable (i.e. ${\cal P}(\mathbb H)\times{\cal B}({\mathbb R}^d)$-measurable), it is enough to prove that this functional is continuous in $\lambda$ (in this case our functional $L_{\delta}$ falls into the class of Carath\'eodory functions), and hence one can conclude  immediately that it is jointly measurable due to \cite[Lemma 4.51]{Aliptrantis}. Thus, the rest of this step focuses on proving that $L_{\delta}$ is continuous in $\lambda$.  To this end, we first remark that $-\lambda^{tr}b(\omega,t)+{1\over{2}}\lambda^{tr}c(\omega,t)\lambda$ is continuous, and we derive
\begin{eqnarray*}
&&-\delta\vert\lambda\vert^2\vert x\vert^2\leq f_{\delta}(\lambda,x)\leq \max({1\over{2(1-\delta)^2}},-\delta-\ln(1-\delta))\vert\lambda\vert^2\vert x\vert^2\ \mbox{on}\ \{\vert x\vert\leq 1\}\\
&&-\delta\vert\lambda\vert\vert x\vert\leq f_{\delta}(\lambda,x)\leq -\ln(1-\delta)\quad \mbox{on}\ \{\vert x\vert> 1\}.
\end{eqnarray*}
Therefore, thanks to the dominated convergence theorem and these inequalities, we deduce that in fact $L_{\delta}$ is continuous in $\lambda$, and the first step is complete.\\
{\bf Step 2:}  Herein, we prove that for any $(\omega,t)\in\Omega\times[0,+\infty)$ and any $\lambda\in{\mathbb R}^d$, $L_{\delta}(\omega,t,\lambda)$ converges to $L(\omega,t,\lambda)$ when $\delta$ goes to one. To this end, we first write 
\begin{eqnarray*}
&&L_{\delta}(\omega,t,\lambda)=\delta \int_{\{\lambda^{tr}x\leq -1\}}\lambda^{tr}h(x)F(dx)-\delta \int_{\{ \lambda^{tr}x\leq -1\}}\lambda^{tr}x I_{\{\vert x\vert>1\}}F(dx)\\
&&-\ln(1-\delta)F\left(\{\lambda^{tr}x\leq -1\}\right)+\int I_{\{\lambda^{tr}x> -1\}}\left(\delta\lambda^{tr}x-\ln(1+\delta\lambda^{tr}x)\right)F(dx).\hskip 1cm
\end{eqnarray*}
Remark that $\int_{\{\lambda^{tr}x\leq -1\}}\lambda^{tr}h(x)F(dx)$ and $\int_{\{ \lambda^{tr}x\leq -1\}}\lambda^{tr}x I_{\{\vert x\vert>1\}}F(dx)$ are well defined and take finite values, while $I_{\{\lambda^{tr}x> -1\}}\left(\delta\lambda^{tr}x-\ln(1+\delta\lambda^{tr}x)\right)$ is nonnegative and increasing in $\delta$. By distinguishing the cases whether $F\left(\{\lambda^{tr}x\leq -1\}\right)$ is null or not, thanks to the convergence monotone theorem, we conclude that $L_{\delta}(\omega,t,\lambda)$ converges to $L(\omega,t,\lambda)$. This ends the second step and the proof of the lemma.
\end{proof}
%%\end{appendix}
%%%%%%%%%%%%%%%%%%%%%%%%%%%%%%%%%%%%%%%%%%%%%%%%%%%%%%%%%%%%%%%%%%%%%%
%%%%%%%%%%%%%%%%%%%%%%%%%%%%%%%%%%%%%%%%%%%%%%%%%%%%%%%%%%%%%%%%%%%%%%%%
%%%%%%%%%%%%%%%%%%%%%%%%%%%%%%%%%%%%%%%%%%%%%%%%%%%%%%%%%%%%%%%%%%%%%%%
{\bf Acknowledgements:} This research is fully supported financially by the
Natural Sciences and Engineering Research Council of Canada, through Grant G121210818. \\ The authors would like to thank  Safa Alsheyab, Ferdoos Alharbi, Jun Deng, Monique Jeanblanc, Youri Kabanov and Michele Vanmalele  for several comments, fruitful discussions on the topic, and/or for providing important and useful references.   

%\acks This research is supported by  the Natural Sciences
%and Engineering Research Council of Canada, through Grant RES0020459. The authors would like to thank  Jun Deng, Youri Kabanov, Martin Schweizer and Michele Vanmalele  for several comments, fruitful discussions on the topic, and/or for providing important and useful references.   
%%%\end{acknowledgements}
%\vspace*{0,5cm}
%%%%%%%\section{Localisation  under $\mathbb G$ versus that of $\mathbb F$}
%%\end{appendix}
%%%%%%%%%%%%%%%%%%%%%%%%%%%%%%%%%%%%%%%%%%%%%%%%%%%%%%%%%%%%%%%%%%%%%%
%%%%%%%%%%%%%%%%%%%%%%%%%%%%%%%%%%%%%%%%%%%%%%%%%%%%%%%%%%%%%%%%%%%%%%%%
%%%%%%%%%%%%%%%%%%%%%%%%%%%%%%%%%%%%%%%%%%%%%%%%%%%%%%%%%%%%%%%%%%%%%%%

\end{document}